\documentclass[conference]{IEEEtran}
\IEEEoverridecommandlockouts
\usepackage[utf8]{inputenc}
\usepackage{hyperref}

\usepackage[labelfont=bf]{caption}

\usepackage{subcaption}

\usepackage{cite}
\usepackage{amsmath,amssymb,amsfonts}
\usepackage{graphicx}
\usepackage{url}

\usepackage{comment}
\usepackage{amsthm}
\newtheorem{definition}{Definition}
\newtheorem{theorem}{Theorem}

\usepackage{xspace}

\usepackage[table]{xcolor}
\usepackage{multirow}
\usepackage{mathtools}
\usepackage{siunitx}

\usepackage{multicol}

\usepackage{algorithm}
\usepackage[noend]{algpseudocode}
\algdef{SE}[DOWHILE]{Do}{doWhile}{\algorithmicdo}[1]{\algorithmicwhile\ #1}

\usepackage{tikz}
\usetikzlibrary{arrows,automata}
\usetikzlibrary{fit,automata,positioning,calc}

\usepackage{wrapfig}

\def\BibTeX{{\rm B\kern-.05em{\sc i\kern-.025em b}\kern-.08em
    T\kern-.1667em\lower.7ex\hbox{E}\kern-.125emX}}

\newcommand{\pre}[1]{\ensuremath{{^\bullet#1}}\xspace}
\newcommand{\post}[1]{\ensuremath{{#1^\bullet}}\xspace}
\newcommand{\preregion}[1]{\ensuremath{{^\circ#1}}\xspace}
\newcommand{\postregion}[1]{\ensuremath{{#1^\circ}}\xspace}

\newcommand{\rgraph}[1]{\ensuremath{{\textit{RG}(#1)}}\xspace}
\newcommand{\es}[1]{\ensuremath{{\textit{ES}(#1)}}\xspace}
\newcommand{\mis}[1]{\ensuremath{{\textit{MIS}(#1)}}\xspace}
\newcommand{\fire}[1]{\ensuremath{{[#1\rangle}}\xspace}

\newcommand{\arc}[3]{\ensuremath{#1\stackrel{#2}{\rightarrow}#3\xspace}}

\newcommand{\pn}[0]{\ensuremath{\textit{PN}}}
\newcommand{\sm}[0]{\ensuremath{\textit{SM}}}

\begin{document}

\title{Decomposition of transition systems into sets of synchronizing state
machines
}

\author{\IEEEauthorblockN{Viktor Teren}
\IEEEauthorblockA{\textit{Department of Computer Science} \\
\textit{Università degli Studi di Verona}\\
Verona, Italy \\
viktor.teren@univr.it}
\and
\IEEEauthorblockN{Jordi Cortadella}
\IEEEauthorblockA{\textit{Department of Computer Science} \\
\textit{Universitat Politècnica de Catalunya}\\
Barcelona, Spain \\
jordi.cortadella@upc.edu}
\and
\IEEEauthorblockN{Tiziano Villa}
\IEEEauthorblockA{\textit{Department of Computer Science} \\
\textit{Università degli Studi di Verona}\\
Verona, Italy \\
tiziano.villa@univr.it}
}

\maketitle

\thispagestyle{plain}
\pagestyle{plain}

\begin{abstract}
Transition systems (TS) and Petri nets (PN) are important models of
computation
ubiquitous in formal methods for modeling systems. An important problem
is how
to extract from a given TS a PN whose reachability graph is equivalent
(with a suitable notion of equivalence) to the original TS.

This paper addresses the decomposition of transition systems into
synchronizing
state machines (SMs), which are a class of Petri nets where each transition
has one incoming and one outgoing arc and all markings have exactly one
token.
This is an important case of the general problem of extracting a PN from
a TS.
The decomposition is based on the theory of regions, and it is shown that
a property of regions called excitation-closure is a sufficient condition
to guarantee the equivalence between the original TS and a decomposition
into SMs.

An efficient algorithm is provided which solves the problem by reducing
its critical steps to the maximal independent set problem (to compute
a minimal  set of irredundant SMs) or to satisfiability (to merge the SMs).
We report experimental results that show a good trade-off between quality
of results vs. computation time.
\end{abstract}


\section{Introduction}
\label{section:introduction}

The decomposition of a transition system (TS) into a synchronous product
of state machines (SMs, Petri nets with exactly one incoming and outgoing edge for every transition) gives an intermediate model between a TS and a Petri net (PN).
The set of SMs may exhibit fewer distributed states and transitions, exploiting the best of both worlds of TSs and PNs, leading to better implementations (e.g., smaller circuits with less power consumption). Furthermore, the decomposition procedure extracts explicitly the system concurrency (a PN feature), which is convenient for system analysis and performance improvement
(see an example in Fig.~\ref{fig:first_example}).

The decomposition of a transition system can be seen from the Petri net perspective as the problem of the coverability by S-components of a Petri net \cite{Kemper1992AnEP,desel1995free,logic_synthesis_2013} or of a connected subnet system\cite[p.~49]{badouel2015petri} (called S-coverability): each S-component is a strongly connected safe SM i.e., SM with only one token, therefore it cannot contain concurrency. 
The only concurrency of the system is featured in the interaction of the S-components. In our paper we present how the theory of regions~\cite{regions} can be used to design a similar procedure starting from a transition system and creating a set of interacting SMs,  
but without building an equivalent Petri net. 
Our approach 
computes a set of minimal regions with the excitation-closure (EC) property of a given TS, and derives from them an irredundant synchronous product 
of interacting SMs. 
Excitation closure guarantees that the regions extracted from the transition system are sufficient to model its behaviour.

\begin{wrapfigure}{r}{0.42\linewidth}
    \vspace{-0.3cm}
    \hspace{-0.5cm}
    \begin{subfigure}{\linewidth}
    \scalebox{0.8}{
    \begin{tikzpicture}[->,>=stealth',shorten >=1pt,auto,node distance=1.2cm,
		semithick,initial text={}]
		\tikzstyle{every state}=[]
		\node [initial,state,inner sep=1pt,minimum size=0pt] (s0) {$s_0$};
		\node [state,inner sep=1pt,minimum size=0pt] (s16) [right of=s0] {$s_{16}$};
		\node [state,inner sep=1pt,minimum size=0pt] (s17) [below of=s16] {$s_{17}$};
		\node [state,inner sep=1pt,minimum size=0pt] (s1) [right of=s16] {$s_1$};
		\node [state,inner sep=1pt,minimum size=0pt] (s19) [below  of=s17] {$s_{19}$};
		\node [state,inner sep=1pt,minimum size=0pt] (s2) [below of=s1] {$s_2$};
		\node [state,inner sep=1pt,minimum size=0pt] (s4) [right of=s1] {$s_4$};
		\node [state,inner sep=1pt,minimum size=0pt] (s3) [below  of=s2] {$s_3$};
		\node [state,inner sep=1pt,minimum size=0pt] (s5) [below  of=s4] {$s_5$};
		\node [state,inner sep=1pt,minimum size=0pt] (s6) [below of=s5] {$s_6$};
		\node [state,inner sep=1pt,minimum size=0pt] (s18) [below left of=s6] {$s_{18}$};
		\node [state,inner sep=1pt,minimum size=0pt] (s7) [below right of=s18] {$s_7$};
		\node [state,inner sep=1pt,minimum size=0pt] (s9) [below  of=s7] {$s_9$};
		\node [state,inner sep=1pt,minimum size=0pt] (s10) [left of=s7] {$s_{10}$};
		\node [state,inner sep=1pt,minimum size=0pt] (s8) [below  of=s9] {$s_8$};
		\node [state,inner sep=1pt,minimum size=0pt] (s12) [below  of=s10] {$s_{12}$};
		\node [state,inner sep=1pt,minimum size=0pt] (s13) [ left of=s10] {$s_{13}$};
		\node [state,inner sep=1pt,minimum size=0pt] (s11) [below  of=s12] {$s_{11}$};
		\node [state,inner sep=1pt,minimum size=0pt] (s15) [below  of=s13] {$s_{15}$};
		\node [state,inner sep=1pt,minimum size=0pt] (s14) [below  of=s15] {$s_{14}$};
		
		\path (s0) edge node {$b-$} (s16)
		        (s16) edge node {$s+$} (s17)
		        (s16) edge node {$r-$} (s1)
		        (s17) edge node {$b+$} (s19)
		        (s17) edge node {$r-$} (s2)
		        (s1) edge node {$s+$} (s2)
		        (s1) edge node {$a+$} (s4)
		        (s19) edge node {$r-$} (s3)
		        (s2) edge node {$b+$} (s3)
		        (s2) edge node {$a+$} (s5)
		        (s4) edge node {$s+$} (s5)
		        (s3) edge node {$a+$} (s6)
		        (s5) edge node {$b+$} (s6)
		        (s6) edge node {$s-$} (s18)
		        (s18) edge node {$b-$} (s7)
		        (s7) edge node {$s+$} (s9)
		        (s7) edge node {$r+$} (s10)
		        (s9) edge node {$b+$} (s8)
		        (s9) edge node {$r+$} (s12)
		        (s10) edge node {$s+$} (s12)
		        (s10) edge node {$a-$} (s13)
		        (s8) edge node {$r+$} (s11)
		        (s12) edge node {$b+$} (s11)
		        (s12) edge node {$a-$} (s15)
		        (s13) edge node {$s+$} (s15)
		        (s11) edge node {$a-$} (s14)
		        (s15) edge node {$b+$} (s14)
		    (s14) edge [bend left=10] node {$s-$} (s0)
		;
\end{tikzpicture}
}\caption{} \label{fig:first_example_a}
\end{subfigure}

\begin{subfigure}{\linewidth}
\hspace{-0.5cm}
\scalebox{0.8}{
    \begin{tikzpicture}[->,>=stealth',shorten >=1pt,auto,node distance=1.2cm,
		semithick,initial text={}]
		\tikzstyle{every state}=[]
		
		\node [initial,state,inner sep=1pt,minimum size=0pt] (r3) {$r_3$};
		\node [state,inner sep=1pt,minimum size=0pt] (r1) [below of=r3] {$r_1$};
		\node [state,inner sep=1pt,minimum size=0pt] (r2) [right of=r1] {$r_2$};
		\node [state,inner sep=1pt,minimum size=0pt] (r0) [above of=r2] {$r_0$};
	
	    \path (r3) edge node {$b-$} (r1)
	        (r1) edge node {$s+$} (r2)
	        (r2) edge node {$b+$} (r0)
	        (r0) edge node {$s-$} (r3)
	    ;
	    
	\end{tikzpicture}
	\begin{tikzpicture}[->,>=stealth',shorten >=1pt,auto,node distance=1.2cm,
		semithick,initial text={}]
		\tikzstyle{every state}=[]

		\node [initial,state,inner sep=1pt,minimum size=0pt] (r3) {$r_3$};
		\node [state,inner sep=1pt,minimum size=0pt] (r0) [below  of=r3] {$r_0$};
		\node [state,inner sep=1pt,minimum size=0pt] (r2) [right of=r0] {$r_2$};
		\node [state,inner sep=1pt,minimum size=0pt] (r1) [above of=r2] {$r_1$};
	
	    \path (r3) edge node {$r-$} (r0)
	        (r0) edge node {$a+$} (r2)
	        (r2) edge node {$r+$} (r1)
	        (r1) edge node {$a-$} (r3)
	    ;
		
		\end{tikzpicture}}
		
		\vspace{-0.3cm}
		\scalebox{0.8}{
	\centering
	\begin{tikzpicture}[->,>=stealth',shorten >=1pt,auto,node distance=1.2cm,
		semithick,initial text={}]
		\tikzstyle{every state}=[]
		\node [initial,state,inner sep=1pt,minimum size=0pt] (r4) {$r_4$};
		\node [state,inner sep=1pt,minimum size=0pt] (r3) [right of=r4] {$r_3$};
		\node [state,inner sep=1pt,minimum size=0pt] (r2) [right of=r3] {$r_2$};
		\node [state,inner sep=1pt,minimum size=0pt] (r0) [above of=r3] {$r_0$};
		\node [state,inner sep=1pt,minimum size=0pt] (r1) [above of=r2] {$r_1$};
	
	   \path (r4) edge node {$b-$} (r3)
	        (r3) edge node {$r+$} (r2)
	        (r3) edge node {$r-$} (r0)
	        (r2) edge node {$a-$} (r1)
	        (r0) edge node {$a+$} (r1)
	        (r1) edge [bend right=99, above] node {$s-$} (r4)
	    ;
		\end{tikzpicture}
	}

\caption{} \label{fig:first_example_b}
\end{subfigure}
\caption{TS derived from an STG\protect\footnotemark (\ref{fig:first_example_a}) and the derived set of synchronizing state machines (\ref{fig:first_example_b}).}
\label{fig:first_example}
\vspace{-0.3cm}
\end{wrapfigure}
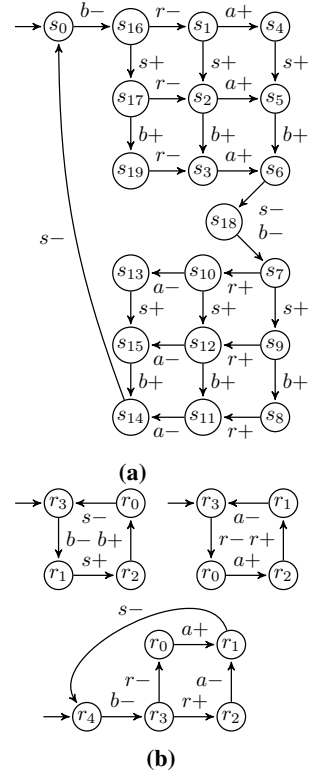

\footnotetext{
A Signal Transition Graph (\textbf{STG}) $G = (V, E)$ is an interpreted subset of marked graphs wherein each transition represents either the rising ($x^+$) or falling ($x^-$) of a signal $x$ which has signal levels high and low. $V$ is the set of transitions and $E$ is the set of edges corresponding to places of the underlying marked graph.
}

The main steps of the decomposition procedure are:
1) computation of all minimal regions of the given TS, 2) generation of a set of SMs with the excitation-closure property, 3) removal of redundant SMs, 4) merging of regions while preserving the excitation-closure property.
The generation of minimal regions is well-known from literature~\cite{regions}.
The generation of SMs with the EC property is reduced to solving instances of maximal independent set, where each solution of MIS yields an SM.
Some of these SMs may be completely redundant, i.e., they can be removed while the remaining partially redundant SMs still satisfy the EC property.
We use a greedy strategy to find a minimal irredundant set of SMs.
The surviving SMs go through a simplification step that merges adjacent regions and removes the edges/labels captured by the merging step. In the extreme case, one can remove all instances of a region except for one SM. The best merging option is selected by encoding both the constraints of the merging operations and the optimization objective as an ILP solvable by SAT and binary search~\cite{boros2002pseudo}, with the goal of keeping the minimum number of labels needed to satisfy the EC property. At the end, the SMs are optimized according to the selected merging operations.

The optimization steps in which the problem is divided may be solved exactly or with heuristics. Experiments have been performed trying various combinations of exact and heuristic algorithms, with the conclusion that the heuristics deliver good results in reasonable computation time.

\subsection{Previous and related work}

In~\cite{kalenkova2014process}, a transition system is decomposed iteratively into an interconnection of $n$ component transitions systems with the objective to extract a Petri net from them. This can be seen as a special case of our problem, because in~\cite{kalenkova2014process} the decomposition allows the extraction of a Petri net, but the decomposed set of transition systems cannot be used as an intermediate model. Their approach is flexible in choosing how to split the original transition system, but it does not provide any minimization algorithm, so that the redundancy due to overlapping states in the component transition systems translates into redundant places of the final Petri net. 
Another method presented in~\cite{de2016mining} is based on the decomposition of transition systems into \textit{``slices"}, where each transition system is separately synthesized into a Petri net, and in case of Petri nets ``hard" to understand the process can be recursively repeated on one or more \textit{``slices"} creating a higher number of smaller PNs. With respect to the aforementioned methods, our approach yields by construction a set of PNs restricted to only SMs and applies to them minimization criteria.

Decomposition plays an important role in process mining~\cite{van2012decomposing, van2013decomposing, verbeek2014decomposed, taibi2019monolithic}, where in most cases the decomposition starts from a Petri net representing the whole behaviour of the system~\cite{van2012decomposing, van2013decomposing, verbeek2014decomposed}.
Instead of creating a PN from event logs we can easily create a transition system~\cite{van2010process, carmona2009divide} and directly decompose it with our algorithm. 

This paper is organized as follows.
Sec.~\ref{section:preliminaries} introduces the background material (including the theory of regions to extract PNs from TSs) and then characterizes the extraction of SMs from TSs. 
Sec.~\ref{section:relations} discusses composition of SMs and contains the main theoretical result that the synchronous product of SMs is bisimilar to the original transition system (proof in the appendix). The procedures to extract the SMs are described in Sec.~\ref{section:decomposition} and exhaustive experiments are reported in Sec.~\ref{section:experiments}, with final conclusions drawn in Sec.~\ref{section:conclusions}.

\section{Preliminaries} \label{section:preliminaries}

\subsection{Transition systems}

\begin{definition}[TS/LTS] \label{def:ts} \cite{regions} A Labeled Transition System (LTS, or simply TS) 
is defined as a 4-tuple ($S$, $E$, $T$, $s_0$) where:

\begin{itemize}
    \item $S$ is a non-empty set of states
    \item $E$ is a set of events/labels
    \item $T \subseteq S \times E \times S $ is a transition relation
    \item $s_{0} \in S$ is an initial state
\end{itemize}
\end{definition}

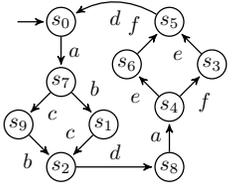
\begin{wrapfigure}{r}{0.33\linewidth}
	\centering
	\scalebox{0.8}{
	\begin{tikzpicture}[->,>=stealth',shorten >=1pt,auto,node distance=1cm,semithick, initial text={}]
	\tikzstyle{every state}=[]
	
	\node [initial,state,inner sep=1pt,minimum size=0pt] (s0) {$s_0$};
	\node [state,inner sep=1pt,minimum size=0pt] (s7) [below of=s0] {$s_7$};
	\node [state,inner sep=1pt,minimum size=0pt] (s9)  [below left of=s7] {$s_9$};
	\node [state,inner sep=1pt,minimum size=0pt] (s1)  [below right of=s7] {$s_1$};
	\node [state,inner sep=1pt,minimum size=0pt] (s2) [below right of=s9] {$s_2$};
	\node [state,inner sep=1pt,minimum size=0pt] (s8) [right of=s2, xshift=0.8cm] {$s_8$};
	\node [state,inner sep=1pt,minimum size=0pt] (s4) [above of=s8] {$s_4$};
	\node [state,inner sep=1pt,minimum size=0pt] (s3) [above right of=s4] {$s_3$};
	\node [state,inner sep=1pt,minimum size=0pt] (s6) [above left of=s4] {$s_6$};
	\node [state,inner sep=1pt,minimum size=0pt] (s5) [above left of=s3] {$s_5$};
	
	\path (s0) edge node {$a$} (s7)
	(s7) edge node {$b$} (s1)
	(s7) edge node {$c$} (s9)
	(s9) edge [below left] node {$b$} (s2)
	(s1) edge [above left] node {$c$} (s2)
	(s2) edge node {$d$} (s8)
	(s8) edge node {$a$} (s4)
	(s4) edge [below right] node {$f$} (s3)
	(s4) edge node {$e$} (s6)
	(s3) edge node {$e$} (s5)
	(s6) edge [above left] node {$f$} (s5)
	(s5) edge [bend right=30] node {$d$} (s0)
	;
	\end{tikzpicture}}
	\caption{Example of transition system.}
	\label{fig:TS}
	\vspace{-0.7cm}
\end{wrapfigure}

Every transition system is supposed to satisfy the following properties:
\begin{itemize}
    \item It does not contain-self loops: $\forall (s, e, s') \in T: s \neq s'$;
    \item Each event has at least one occurrence: 
     $\forall e \in E: \exists(s,e,s') \in T$;
    \item Every state is reachable from the initial state: \mbox{$\forall s \in S: s_{0} \rightarrow^{*} s$;}
\end{itemize}

\begin{itemize}
    \item It is deterministic: 
    for each state there is at most one successor state reachable with label $e$.
\end{itemize}

An example of a transition system can be seen in Fig.~\ref{fig:TS}.

\begin{definition}[Isomorphism] \label{def:isomorphism}
Two transition systems \mbox{$\textit{TS}_1 = (S_1, E, T_1, s_{0,1})$} and 
\mbox{$\textit{TS}_2 = (S_2, E, T_2, s_{0,2})$} are said to be isomorphic (or that there is an isomorphism between $\textit{TS}_1$ and $\textit{TS}_2$) if there is a bijection \mbox{$b_S: S_1 \rightarrow S_2$}, such that:
\begin{itemize}
    \item $b_S(s_{0,1}) = s_{0,2}$
    \item $\forall (s,e,s') \in T_1:~ (b_S(s),e,b_S(s'))\in T_2$
    \item $\forall (s,e,s') \in T_2:~ (b_S^{-1}(s),e,b_S^{-1}(s')) \in T_1$.
\end{itemize}
\end{definition}

\begin{definition}[Bisimulation] \label{def:bisimulation}
Given two transition systems \mbox{$\textit{TS}_1 = (S_1, E, T_1,$} $s_{0,1})$ and \mbox{$\textit{TS}_2 = (S_2, E, T_2, s_{0,2})$}, a binary relation $B \subseteq S_1 \times S_2$ is a bisimulation, denoted by \mbox{$\textit{TS}_1 \sim_B \textit{TS}_2$,} if $(s_{0,1}, s_{0,2}) \in B$ and if whenever $(p, q) \in B$ with $p \in S_1$ and $q \in S_2$:
\begin{itemize}
    \item \mbox{$\forall (p, e, p') \in T_1:~\exists q' \in S_2$} such that
    \mbox{$(q, e, q') \in T_2$} and \mbox{$(p', q') \in B$}
    \item \mbox{$\forall (q, e, q') \in T_2:~ \exists p' \in S_1$} such that
    \mbox{$(p, e, p') \in T_1$} and \mbox{$(p',q') \in B$}.
\end{itemize}
Two TSs are said to be bisimilar if there is a bisimulation between them.
\end{definition}

The operation $\textit{Ac}$ deletes from a TS all the states that are not reachable or accessible from the initial state and all transitions attached to them.

\begin{definition}[Synchronous product] \label{def:parallel_composition}
Given two transition systems \mbox{$\textit{TS}_1 = (S_1,$} $E_1, T_1, s_{0,1})$ and $\textit{TS}_2 = (S_2, E_2, T_2, s_{0,2})$, the synchronous product 
is defined as \mbox{$\textit{TS}_1 || \textit{TS}_2 = \textit{Ac}(S, E_1 \cup E_2, T, (s_{0,1}, s_{0,2}))$} where $S \subseteq S_1 \times S_2$, $(s_{0,1}, s_{0,2}) \in S$, $T \subseteq (S_1 \times S_2) \times E \times (S_1 \times S_2)$ is defined as follows:

\begin{itemize}
	\item  if $a \in E_1 \cap E_2$, 
			$(s_1, a, s_1') \in T_1$ and $(s_2,a,s_2') \in T_2$ then \mbox{$((s_1,s_2), a, (s_1',s_2')) \in T$,}
	\item  if $a \in E_1$, $a \notin E_2$ and 
			 $(s_1, a, s_1') \in T_1$ then $((s_1,s_2), a, (s_1',s_2)) \in T$,
	\item  if $a \notin E_1$, $a \in E_2$ and 
			 $(s_2, a, s_2') \in T_2$ then $((s_1,s_2), a, (s_1,s_2')) \in T$,
    \item nothing else belongs to $T$.
\end{itemize}
\end{definition}
The synchronous product is associative, so we can define the product of a collection of $n$ TSs: $\textit{TS}_1 || \textit{TS}_2 || \dots || \textit{TS}_n = ((\textit{TS}_1 || \textit{TS}_2) \dots) || \textit{TS}_n$; as an alternative, we can extend directly the previous definition to more than two TSs.

\subsection{Petri Nets}

We assume the reader to be familiar with Petri nets. We refer to~\cite{murata1989petri} for a deeper insight on the concepts used in this work. This section introduces the nomenclature related to Petri nets used along the paper.

In this work we will only deal with safe Petri nets, i.e., nets whose places do not contain more than one token in any reachable marking. For this reason, we will model markings as sets of places.

\begin{definition}[Ordinary Petri Net] \label{def:PN} \cite{murata1989petri}
	An ordinary Petri net is a 4-tuple, \hbox{$\pn=(P, T,$} $F, M_0)$ where:
	\begin{itemize}
		\item $P = \{p_1, p_2, ..., p_m\}$ is a finite set of places,
		\item $T = \{t_1, t_2, ..., t_n\}$ is a finite set of transitions,
		\item $F \subseteq (P \times T) \cup (T \times P)$ is a set of arcs (flow relation),
		\item $M_0$ is an initial marking,
		\item $P \cap T = \emptyset$ and $P \cup T \neq \emptyset$.
	\end{itemize}
	A Petri net structure $N=(P,T,F)$ without any specific initial marking is denoted by $N$.
	A Petri net with an initial marking $M_0$ is denoted by $(N, M_0)$.
	
	For any \mbox{$x\in P\cup T$}, then \mbox{$\pre{x}=\{ y | (y,x)\in F\}$}. Similarly, \mbox{$\post{x}=\{ y \mid (x,y)\in F\}$}.
\end{definition}

\begin{definition}[Reachability graph] \label{def:reachability_graph} \cite[p. 20]{badouel2015petri}    
Given a safe Petri net \mbox{$N = (P,$} $T, F, M_0)$, the reachability graph of $N$ is the transition system 
\mbox{$\rgraph{N} = ([M_0\rangle, T,$} $\Delta, M_0)$ defined by 
\mbox{$(M, t, M') \in \Delta$} if $M \in [M_0\rangle$ and \mbox{$M[t\rangle M'$}.
\end{definition}

\begin{definition}[State Machine, SM] \label{def:SM} \cite{murata1989petri}
	A state machine is an ordinary Petri net $N = (P, T, F, M_0)$ such that for every transition $t \in T$, $|\pre{t}|=|\post{t}|=1$, i.e., it has exactly one incoming and one outgoing edge. In a safe State Machine it also holds that $|M_0| = 1$.
\end{definition}

It has been observed in \cite[p. 49]{badouel2015petri} that a state machine \mbox{$M = (P, T, F, M_0)$} can be interpreted as a transition system \mbox{$\textit{TS} = (P, T, \Delta, s_0)$}, where the places correspond to the states,
the transitions to the events, $s_0$ corresponds to the unique marked initial place, and $(p,t,p') \in \Delta$ iff $\pre{t} = \{p\}$ and $\post{t} = \{p'\}$ (in a SM by definition $|\pre{t}|=|\post{t}|=1$).
Therefore the reachability graph of $M$ is isomorphic to the transition system $\textit{TS}$, i.e., $\rgraph{M}$ is isomorphic to $\textit{TS}$.

In this paper we consider sets of synchronizing SMs.

\subsection{From LTS to Petri nets by regions}

In this paper we propose a procedure for the decomposition of Transition Systems based on the theory of regions (from \cite{regions}).
A region is a subset of states in which all the transitions under the same event have the same relation with the region: either all entering, or all exiting, or some completely inside and some completely outside the region.
\begin{definition}[Region] \label{def:reg}
	Given a $\textit{TS} = (S, E,$ $T, s_{0})$, a region is defined as a non-empty set of states  $r \subsetneq S$ such that the following properties hold for each event $e \in E$: 
\vspace{-0.2cm}
\begin{eqnarray*}
enter(e,r) & \implies & \neg in(e,r) \land \neg out(e,r) \land \neg exit(e,r) \\
exit(e,r) & \implies & \neg in(e,r) \land \neg out(e,r) \land \neg enter(e,r)
\end{eqnarray*}

\vspace{-0.2cm}
\textit{where}
\vspace{-0.1cm}
\[
\begin{array}{l}
\left.
\begin{array}{rcl}
in(e,r) & \equiv & \exists (s,e,s') \in T : s,s' \in r \\
~~out(e,r) & \equiv & \exists (s,e,s') \in T : s,s' \notin r
\end{array}
\quad \right\} \textit{no\_cross} \\
\begin{array}{rcl}
enter(e,r) & \equiv & \exists (s,e,s') \in T : s \notin r \land s' \in r \\	
exit(e,r) & \equiv&  \exists (s,e,s') \in T : s \in r \land s' \notin r
\end{array}
\end{array}
\]
\end{definition}
\begin{definition}[Minimal region] \label{def:minimal_region}
	A region $r$ is called \textit{minimal} if there is no other region $r'$ strictly contained in $r$ ($\nexists r' \mid r' \subset r$).
\end{definition}

\begin{definition}[Pre-region (Post-region)] \label{def:pre/post_region}
	A region $r$ is a pre-region (post-region) of an event $e$ if there is a transition labeled with $e$ which exits from $r$ (enters into $r$). The set of all pre-regions (post-regions) of the event $e$ is denoted by $\preregion{e}$ ($\postregion{e}$).
\end{definition}

By definition if $r \in \preregion{e}$ ($r \in \postregion{e}$) all the transitions labeled with $e$ are exiting from $r$ (entering into $r$), furthermore, if the transition system is strongly connected all the regions are also pre-regions of some event.

\begin{definition}[Excitation set / Switching set] \label{def:ES} \label{def:SS}
	The \emph{excitation (switching) set} of event $e$, $\es{e}$ ($\textit{SS}(e)$), is the maximal set of states such that for every $s \in \es{e}$ ($s \in \textit{SS}(e)$) there is a transition  $\arc{s}{e}{}$ ($\arc{}{e}{s}$).
\end{definition}

\begin{definition}[Excitation-closed Transition System (ECTS)] \label{def:excitation_closure}
	A TS with the set of labels $E$ and the pre-regions $\preregion{e}$ is an ECTS if the following conditions are satisfied:
	\begin{itemize}
		\item Excitation closure: $\forall e \in E: \bigcap_{r \in \preregion{e}} r=\es{e}$
		\item Event effectiveness: $\forall e \in E:$ $\preregion{e} \neq \emptyset$
	\end{itemize}
\end{definition}

If the initial TS does not satisfy the excitation closure (EC) or event effectiveness property, \textit{label splitting}\cite{regions} can be performed to obtain an ECTS.

The EC property also ensures that if two states, $s_1$ and $s_2$ cannot be \emph{separated} by any region, i.e., there is no minimal region $r$ such that $s_1\in r$ and $s_2 \not\in r$, then $s_1$ and $s_2$ are bisimilar.

The synthesis of a Petri net from an ECTS, proposed in~\cite{regions}, can be summarized by the following steps:

\begin{enumerate}
\item \emph{Generation of all minimal regions}.\\
All the excitation sets are expanded until they become
regions, i.e. all events satisfy one of the {\it enter/exit/no\_cross} conditions with respect to the regions. The non-minimal regions can be removed by comparing them to the other regions.

\item \emph{Removal of redundant regions}.\\
Some minimal regions may be redundant, meaning that they can be removed while the 
excitation-closure property still holds.

\item \emph{Merging minimal regions}.\\
In order to obtain a place-minimal PN, subsets of disjoint minimal regions can be merged into non-minimal regions, thus reducing the number of places. This merging must preserve the excitation-closure of the final set of regions.
\end{enumerate}

\subsection{From LTS to SMs by regions}

We now show how to decompose an ECTS into a set of synchronizing SMs.

From the set of all minimal regions obtained from an ECTS we can extract subsets of regions representing state machines. 
A set of regions $R$ represents a state machine if $R$ covers all the states $S$ of the transition system and all the regions are disjoint, i.e.:
\[
\forall r \in R, \nexists r' \in R:  r \cap r' \neq \emptyset
\qquad   \wedge \qquad
\forall s \in S, \exists r \in R: s \in r
\]
Given a set of regions satisfying the previous properties we obtain a state machine whose places correspond to the regions, with a transition $\arc{r_i}{e}{r_j}$ when $r_i$ and $r_j$ are pre- and post-regions of $e$, respectively. 
Since the regions of an SM are disjoint, each derived SM has only one marked place, which corresponds to the regions that cover the initial state.
Notice that \emph{only} the events that cross some region appear in the SM.
Notice also that the reachability property of the original TS is inherited by the SMs obtained by this construction.
\begin{theorem} \label{thm:regions_to_SM}
	Given an ECTS $\textit{TS}=(S,E,T,s_0)$ and the set of all its minimal regions, a subset of regions $R$ represents an SM if and only if the set covers all the states of $\textit{TS}$ and all its regions are pairwise disjoint.
\end{theorem}
\begin{proof}
The proof is based on the fact that every event appearing in one SM can only have one pre-region and one post-region in the SM. Therefore, each event has one incoming and one outgoing edge in the SM.

Given a collection $R$ of disjoint regions that cover all states of TS, each element $r_i \in R$ has entering, exiting and no-crossing events.
We claim that:
\begin{enumerate}
\item If event $e$ exits (enters) region $r_i \in R$ it cannot exit (enter) region \mbox{$r_j \in R, j \neq i$}.
\item If event $e$ exits (enters) region $r_i \in R$, there must be a region $r_j \in R, j \neq i,$ such that event $e$ enters (exits) \mbox{$r_j \in R, j \neq i$}.
\end{enumerate}
We prove the first claim.
Given a region $r_i$ with $e$ as exiting event, there cannot be another region $r_j$ such that $e$ is an exiting event also for $r_j$.
Otherwise, i.e. if $r_i \in\ \preregion{e}$ and $r_j \in\ \preregion{e}$, $j \neq i$, there are two transitions $\arc{s_a}{e}{s_b}$ and $\arc{s_c}{e}{s_d}$ with $s_a \in r_i$ and $s_c \in r_j$. There are two options for $s_b$: either it is inside or outside $r_j$, i.e., $s_b \in r_j$ or $s_b \not\in r_j$, which means that $e$ would either be entering or no-crossing for $r_j$, contradicting that by construction $r_j$ is a region with $e$ as an exiting arc. The same reasoning applies when $e$ is an entering event.

We prove the second claim: if event $e$ appears as exiting (entering) event of $r_i \in R$, it must appear as entering (exiting) event of $r_j \in R$. Indeed, suppose that $r_i \in\ \preregion{e}$, then there is a transition $\arc{s_a}{e}{s_b}$ with $s_a \in r_i$ and $s_b \not\in r_i$, but then there must exist a region $r_j \in R, j \neq i$, such that $s_b \in r_j$, because the union of the regions in $R$ covers all the states of the original $\textit{TS}$, and so $r_j \in \postregion{e}$. The case $r_i \in \postregion{e}$ is proved similarly.

Notice that we use also the fact that in our definition of TS we rule out self-loops.
\end{proof}

The property of excitation closure can be inherited by the SMs, as stated in the following definition.
\begin{definition}[Excitation-closed set of State Machines derived from an ECTS] \label{def:ec_set}
	Given a set of SMs $S$ derived from an ECTS $\textit{TS}$, the set of all regions $R$ of $S$, the set of labels $E$ of $\textit{TS}$, the sets of pre-regions $\preregion{e}$ of the $\textit{TS}$ for all $e \in E$:
	
	$S$ is excitation-closed with respect to the regions of $\textit{TS}$ if the following condition is satisfied:
	\begin{itemize}
		\item EC: 
		$\forall e \in E: \bigcap_{r \in (\preregion{e} \cap R)} r=\es{e}$
		\item Event effectiveness: $\forall e \in E: \exists r \in R \mid r \in \preregion{e}$
	\end{itemize}
\end{definition}

\section{The decomposition algorithm} \label{section:decomposition}

The first step to decompose a transition system is to enumerate all the minimal regions of the original TS. Each collection of disjoint regions covering all the states of the TS represents a state machine, such that the regions are mapped to places of the SM, i.e., each such SM includes a subset of regions of the original TS and represents only the behavior related to the transitions entering into these regions or exiting from them (instead, internal and external events are missing). 

The example in Sec. \ref{section:relations} shows also that we do not need all the SMs to reconstruct the original LTS, so the question is how many of them we need and which is the ``best'' (in some sense) subset of SMs sufficient to represent the given LTS. 
Therefore we may set up a search to obtain a subset of SMs, which are excitation-closed and cover all events, to yield a composition equivalent to the original TS. 
An easy strategy to guarantee the complete coverage of all events is to add new SMs until all regions are used. However, the resulting collection of SMs may contain completely or partially redundant SMs (see Secs.~\ref{section:greedy_removal} and \ref{section:merge}), which can be removed exactly or greedily by verifying the excitation-closure property. 
Moreover, the size of the selected SMs can be reduced through removing redundant labels by merging regions.
As a summary, Algorithm~\ref{alg:decomposition} shows a preliminary sketch of the decomposition procedure.

\begin{wrapfigure}{l}{0.65\linewidth}
	\vspace{-0.6cm}
	\centering
	\scalebox{0.8}{
	\begin{minipage}{1.25\linewidth}
	\begin{algorithm}[H]
		\caption{Decomposition}
		\label{alg:decomposition}
		\begin{algorithmic}[1]
			\Require~~ An ECTS
			\Ensure~~ A minimal set of interacting SMs
			\State {Computation of all minimal regions
			}
			\State {Generation of a set of SMs with \textit{EC} property}
			\State {Removal of redundant SMs}
			\State{Merge of regions preserving the \textit{EC} property}
			
		\end{algorithmic}
	\end{algorithm}
	\end{minipage}
	}
	\vspace{-0.4cm}
\end{wrapfigure}
The first step of the algorithm can be achieved by a greedy algorithm from the literature, which checks minimality while creating regions \cite{regions}\cite[p.~103]{badouel2015petri}\cite{cortadella1997petrify}. 

The second step of the decomposition algorithm is performed by reducing it to an instance of maximal independent set (MIS)\footnote{
	Given an undirected graph $G=(V,E)$, an \textbf{independent set} is a subset of nodes $U \subseteq V$ such that no two nodes in $U$ are adjacent.  
An independent set is maximal if no node can be added without violating independence.
}, and by calling a MIS solver on the graph whose vertices correspond to the minimal regions with edges which connect intersecting regions. Each \textit{maximal independent set} of the aforementioned graph corresponds to a set of disjoint regions that define an SM. 

A greedy algorithm is used for the computation of the third step: starting from the SM with the highest number of regions, one removes each SM whose removal does not invalidate the ECTS properties.

The last step of merging is reduced to a SAT instance, by encoding all the regions of each SM and also the events implied by the presence of one or more regions. Solving this SAT instance by a SAT solver, the number of labels can be minimized by merging the regions which occur multiple times in different SMs.

\subsection{Generation of a set of SMs with excitation closure}

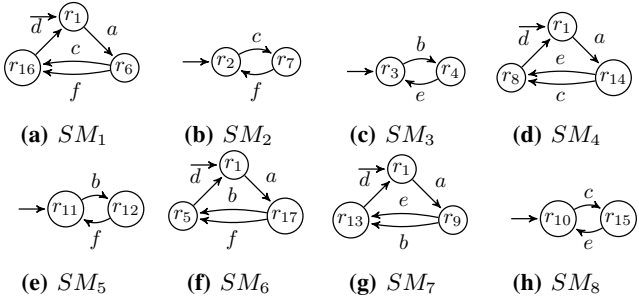
\begin{figure}
	\centering
	\begin{subfigure}[t]{.1\textwidth}
		\centering
		\scalebox{0.8}{
		\begin{tikzpicture}[->,>=stealth',shorten >=1pt,auto,node distance=1.2cm,semithick, initial text={}]
		\tikzstyle{every state}=[]
		
		\node [initial,state,inner sep=1pt,minimum size=0pt] (r1) {$r_1$};
		\node [state,inner sep=1pt,minimum size=0pt] (r6)  [below right of=r1] {$r_6$};
		\node [state,inner sep=1pt,minimum size=0pt] (r16)  [below left of=r1] {$r_{16}$};
		
		\path (r1) edge node {$a$} (r6)
		(r16) edge node {$d$} (r1)
		(r6) edge [bend left=10] node {$f$} (r16)
		(r6) edge [bend right=10, above] node {$c$} (r16)
		;
		\end{tikzpicture}}
		\caption{$SM_1$}
	\end{subfigure}
	~
	\begin{subfigure}[t]{.1\textwidth}
		\centering
		\scalebox{0.8}{
		\begin{tikzpicture}[->,>=stealth',shorten >=1pt,auto,node distance=1cm,semithick, initial text={}]
		\tikzstyle{every state}=[]
		
		\node [initial, state,inner sep=1pt,minimum size=0pt] (r2) {$r_2$};
		\node [state,inner sep=1pt,minimum size=0pt] (r7) [right of=r2] {$r_7$};
		
		\path (r2) edge [bend left] node {$c$} (r7)
		(r7) edge [bend left] node {$f$} (r2)
		;
		\end{tikzpicture}}
		\caption{$SM_2$}
	\end{subfigure}
	~
	\begin{subfigure}[t]{.1\textwidth}
		\centering
		\scalebox{0.8}{
		\begin{tikzpicture}[->,>=stealth',shorten >=1pt,auto,node distance=1cm,semithick, initial text={}]
		\tikzstyle{every state}=[]
		
		\node [initial, state,inner sep=1pt,minimum size=0pt] (r3) {$r_3$};
		\node [state,inner sep=1pt,minimum size=0pt] (r4) [right of=r3] {$r_4$};
		
		\path (r3) edge [bend left] node {$b$} (r4)
		(r4) edge [bend left] node {$e$} (r3)
		;
		\end{tikzpicture}}
		\caption{$SM_3$}
	\end{subfigure}
	~
	\begin{subfigure}[t]{.1\textwidth}
		\centering
		\scalebox{0.8}{
		\begin{tikzpicture}[->,>=stealth',shorten >=1pt,auto,node distance=1.2cm,semithick, initial text={}]
		\tikzstyle{every state}=[]
		
		\node [initial,state,inner sep=1pt,minimum size=0pt] (r1) {$r_1$};
		\node [state,inner sep=1pt,minimum size=0pt] (r8)  [below left of=r1] {$r_8$};
		\node [state,inner sep=1pt,minimum size=0pt] (r14)  [below right of=r1] {$r_{14}$};
		
		\path (r1) edge node {$a$} (r14)
		(r8) edge node {$d$} (r1)
		(r14) edge [bend left=10] node {$c$} (r8)
		(r14) edge [bend right=10, above] node {$e$} (r8)
		;
		\end{tikzpicture}}
		\caption{$SM_4$}
	\end{subfigure}
	
	\begin{subfigure}[t]{.1\textwidth}
		\centering
		\scalebox{0.8}{
		\begin{tikzpicture}[->,>=stealth',shorten >=1pt,auto,node distance=1cm,semithick, initial text={}]
		\tikzstyle{every state}=[]
		
		\node [initial, state,inner sep=1pt,minimum size=0pt] (r11) {$r_{11}$};
		\node [state,inner sep=1pt,minimum size=0pt] (r12) [right of=r11] {$r_{12}$};
		
		\path (r11) edge [bend left] node {$b$} (r12)
		(r12) edge [bend left] node {$f$} (r11)
		;
		\end{tikzpicture}}
		\caption{$SM_5$}
	\end{subfigure}
    ~
	\begin{subfigure}[t]{.1\textwidth}
		\centering
		\scalebox{0.8}{
		\begin{tikzpicture}[->,>=stealth',shorten >=1pt,auto,node distance=1.2cm,semithick, initial text={}]
		\tikzstyle{every state}=[]
		
		\node [initial,state,inner sep=1pt,minimum size=0pt] (r1) {$r_1$};
		\node [state,inner sep=1pt,minimum size=0pt] (r5)  [below left of=r1] {$r_5$};
		\node [state,inner sep=1pt,minimum size=0pt] (r17)  [below right of=r1] {$r_{17}$};
		
		\path (r1) edge node {$a$} (r17)
		(r17) edge [bend right=10, above] node {$b$} (r5)
		(r17) edge [bend left=10] node {$f$} (r5)
		(r5) edge node {$d$} (r1)
		;
		\end{tikzpicture}}
		\caption{$SM_6$}
	\end{subfigure}
	~
	\begin{subfigure}[t]{.1\textwidth}
		\centering
		\scalebox{0.8}{
		\begin{tikzpicture}[->,>=stealth',shorten >=1pt,auto,node distance=1.2cm,semithick, initial text={}]
		\tikzstyle{every state}=[]
		
		\node [initial,state,inner sep=1pt,minimum size=0pt] (r1) {$r_1$};
		\node [state,inner sep=1pt,minimum size=0pt] (r9)  [below right of=r1] {$r_9$};
		\node [state,inner sep=1pt,minimum size=0pt] (r13)  [below left of=r1] {$r_{13}$};
		
		\path (r1) edge node {$a$} (r9)
		(r13) edge node {$d$} (r1)
		(r9) edge [bend left=10] node {$b$} (r13)
		(r9) edge [bend right=10, above] node {$e$} (r13)
		;
		\end{tikzpicture}}
		\caption{$SM_7$}
	\end{subfigure}
	~
	\begin{subfigure}[t]{.1\textwidth}
		\centering
		\scalebox{0.8}{
		\begin{tikzpicture}[->,>=stealth',shorten >=1pt,auto,node distance=1cm,semithick, initial text={}]
		\tikzstyle{every state}=[]
		
		\node [initial, state,inner sep=1pt,minimum size=0pt] (r10) {$r_{10}$};
		\node [state,inner sep=1pt,minimum size=0pt] (r15) [right of=r10] {$r_{15}$};
		\path (r10) edge [bend left] node {$c$} (r15)
		(r15) edge [bend left] node {$e$} (r10)
		;
		\end{tikzpicture}}
		\caption{$SM_8$}
	\end{subfigure}
	\caption{All SMs created from TS in Fig.~\ref{fig:TS}.}
	\label{fig:complete_decomposition_TS}
	\vspace{-0.5cm}
\end{figure}

Given a set of minimal regions of an excitation-closed TS,
Algorithm~\ref{alg:EC_set_generation} returns an excitation-closed set
of SMs, by associating sets of non-overlapping regions to SMs
as mentioned below.
Notice that in Def.~\ref{def:ec_set} we extended
Def.~\ref{def:excitation_closure} of an excitation-closed transition
system (ECTS) to an excitation-closed set of SMs, 
by requiring that the two properties of excitation-closure and 
event-effectiveness hold on the union of regions underlying the SMs.

\begin{figure}
\vspace{-0.2cm}
\centering
\scalebox{0.8}{
\begin{minipage}{1.2\linewidth}
\begin{algorithm}[H]
\caption{Generation of excitation-closed set of SMs}
\label{alg:EC_set_generation}
\begin{algorithmic}[1]
\Require~~ Set of minimal regions of an ECTS 
\Ensure~~ An excitation-closed set of SMs

\State{Create the graph $G$ where each node is a region and there is an edge between intersecting regions} \label{line:graph_creation}
\State{$G_{0}$ $\leftarrow G$} \label{line:g0} 
\State{$M$ $\leftarrow \emptyset$}, $F$ $\leftarrow \emptyset$ \label{line:empty_MIS_F} 
\Do
    \State{Compute $m = MIS(G)$} \label{line:first_MIS} 
    \State{$M \leftarrow M \cup \{ m \}$} \label{line:store_MIS} 
    \State{$G \leftarrow G \setminus M$} \label{line:vertices_removal} 
\doWhile{$G \neq \emptyset$} 
\For{$m \in M$} \label{line:start_sm_creation}
	\State{Compute $\tilde{m} = MIS(G_0)$ with the constraint $\tilde{m} \supseteq m$}
        \State{Build state machine $\tilde{sm}$ induced by set of regions $\tilde{m}$} \label{line:store_SM} 
	\State{$F \leftarrow F \cup \{ \tilde{sm} \}$} \label{line:store_F} 
\EndFor
\State{\Return {$F$}}
\end{algorithmic}
\end{algorithm}
\end{minipage}}
\vspace{-0.6cm}
\end{figure}
Initially, Algorithm~\ref{alg:EC_set_generation} converts the minimal regions of the TS into a graph $G$, where intersecting regions define edges between the nodes of $G$ (line \ref{line:graph_creation}). 
As long as $G$ is not empty, the search of the maximal independent sets is performed on it by invoking the procedure MIS on $G$ ($\mis{G}$, line \ref{line:first_MIS}), storing the results in $M$ (line \ref{line:store_MIS}) and removing the vertices selected at each iteration (line \ref{line:vertices_removal}). In this way, each vertex will be included in one MIS solution.
Notice that the maximal independent sets computed after the first one are not maximal with respect to the original graph $G_0$, because the MIS procedure is run on a subgraph of $G_0$ without the previously selected nodes. 
To be sure 
that we obtain maximal independent sets with respect to the original $G_0$, we expand to maximality the independent sets in $M$, by invoking the MIS procedure on each independent set $m \in M$ constrained to obtain a maximal independent set $\tilde{m} \supset m$ on $G_0$ (from line \ref{line:start_sm_creation}).
Then from the maximal independent sets we obtain the induced state machines to be stored in $F$ (from line \ref{line:store_F}). The motivation of this step to enlarge the independent sets is to increase the number of regions for each SM, in order to widen the space of solutions for the successive optimizations of redundancy elimination and merging. 
The set of SMs derived from Algorithm~\ref{alg:EC_set_generation} satisfies the EC and event-effectiveness properties because by construction each region is included in at least one independent set.

Fig.~\ref{fig:complete_decomposition_TS} shows the resultant SMs derived from the TS in Fig.~\ref{fig:TS}

\subsection{Removal of the redundant SMs} \label{section:greedy_removal}

\begin{table*}[ht]
		\resizebox{\textwidth}{!}{
		\begin{tabular}{|cl|rrrr|S[table-format=4.2]S[table-format=2.2]S[table-format=2.2]S[table-format=2.2]S[table-format=5.2]|S[table-format=2.2]S[table-format=2.2]S[table-format=2.2]S[table-format=2.2]|} \hline
			&\multicolumn{1}{c|}{\bf Input} & \multicolumn{1}{c}{\bf States} & \multicolumn{1}{c}{\bf Transitions} & \multicolumn{1}{c}{\bf Events} & \multicolumn{1}{c|}{\bf Regions} & \multicolumn{1}{c}{\begin{tabular}[c]{@{}c@{}}\bf Time\\ \bf region\\ \bf generation \\ \bf {[}s{]}\end{tabular}} & \multicolumn{1}{c}{\begin{tabular}[c]{@{}c@{}}\bf Time\\ \bf \bf decomposition \\ \bf {[}s{]}\end{tabular}} & \multicolumn{1}{c}{\begin{tabular}[c]{@{}c@{}}\bf Time\\ \bf Greedy \\ \bf {[}s{]}\end{tabular}} & \multicolumn{1}{c}{\begin{tabular}[c]{@{}c@{}}\bf Time \\ \bf Merge \\ \bf {[}s{]}\end{tabular}} & \multicolumn{1}{c|}{\begin{tabular}[c]{@{}c@{}}\bf Total \\ \bf time \\ \bf {[}s{]}\end{tabular}} & \multicolumn{1}{c}{\begin{tabular}[c]{@{}c@{}}\bf Time\\ \bf region\\ \bf generation\\ \bf {[}\%{]}\end{tabular}} & \multicolumn{1}{c}{\begin{tabular}[c]{@{}c@{}}\bf Time\\ \bf decomposition \\ \bf {[}\%{]}\end{tabular}} & \multicolumn{1}{c}{\begin{tabular}[c]{@{}c@{}}\bf Time\\ \bf Greedy \\ {[}\bf \%{]}\end{tabular}} & \multicolumn{1}{c|}{\begin{tabular}[c]{@{}c@{}}\bf Time \\ \bf Merge \\ {[}\bf \%{]}\end{tabular}}\\ \hline
			\parbox[t]{2mm}{\multirow{16}{*}{\rotatebox[origin=c]{90}{``Small-sized" set}}} & alloc-outbound & 17     & 18          & 14              & 15      & 0.00                      & 0.25                   & 0.00           & 0.06           & 0.31             & 0.36     & 80.36 & 0.07 & 19.21                    \\
			& clock          & 10     & 10          & 4               & 11      & 0.01                      & 0.20                   & 0.00            & 0.03           & 0.24             & 3.02               & 85.71 & 0.13 & 11.14          \\
			& dff            & 20     & 24          & 7               & 20      & 0.29                      & 0.20                   & 0.00            & 0.77           & 1.27            & 23.28                    & 15.50 & 0.08 & 61.14    \\
			& espinalt       & 27     & 31          & 20              & 23      & 0.00                  & 0.21                    & 0.00            & 0.49           & 0.70             & 0.37                     & 29.54 & 0.07 & 70.02    \\
			& fair\_arb      & 13     & 20          & 8               & 11      & 0.02                      & 0.20                   & 0.00            & 0.03           & 0.25            & 8.80                 & 80.41 & 0.04 & 10.74        \\
			& future         & 36     & 44          & 16              & 19      & 0.03                      & 0.21               & 0.00            & 0.11          & 0.35            & 9.40                   & 60.35 & 0.20 & 30.05      \\
			& intel\_div3    & 8      & 8           & 4               & 8       & 0.00                      & 0.23                   & 0.00            & 0.01           & 0.24             & 0.75                & 94.73 & 0.04 & 4.48         \\
			& intel\_edge    & 28     & 36          & 6               & 27      & 1.60                      & 0.20                   & 0.00            & 1.30           & 3.11             & 51.58                & 6.41  & 0.14 & 41.86        \\
			& isend          & 53     & 66          & 15              & 128     & 57.67                     & 0.31                   & 0.32            & 1.04           & 59.33 & 97.21        & 0.51  & 0.53 & 1.75                 \\
			& lin\_edac93    & 20     & 28          & 8               & 10      & 0.00                      & 0.19                   & 0.00            & 0.01           & 0.21             & 1.16             & 93.38 & 0.10 & 5.36            \\
			& master-read    & 8932   & 36       & 26              & 33      & 6.71                      & 0.53                   & 0.12            & 1.03           & 8.39             & 80.00         & 6.28  & 1.45 & 12.28               \\
			& pe-rcv-ifc     & 46     & 62          & 16              & 7       & 8.80                        & 0.19                     & 0.00              & 1.21             & 10.21              & 86.20                  & 1.90  & 0.01 & 11.90      \\
			& pulse          & 12     & 12          & 6               & 33      & 0.00                      & 0.19                   & 0.00            & 0.01           & 0.19             & 0.36           & 96.62 & 0.05 & 2.96              \\
			& rcv-setup      & 14     & 17          & 10              & 11      & 0.00                      & 0.19                   & 0.00            & 0.04           & 0.23             & 1.41                & 81.36 & 0.09 & 17.14         \\
			& vme\_read      & 255    & 668         & 26              & 44      & 0.53                       & 0.20                   & 0.01            & 15.17          & 15.91            & 3.36            & 1.23  & 0.04 & 95.37             \\
			& vme\_write     & 821    & 2907        & 30              & 51      & 2.78                      & 0.24                   & 0.03            & 30.03          & 33.08            & 8.39         & 0.73  & 0.10 & 90.77                \\ \hline
			\parbox[t]{2mm}{\multirow{19}{*}{\rotatebox[origin=c]{90}{``Large-sized" set}}} 
			& art\_3\_10     & 32000  & 93200       & 60              & 64      & 154.96                   & 2.02                   & 0.07            & 1.18           & 158.23          & 97.93          & 1.28  & 0.04 & 0.75              \\
			& art\_3\_11     & 42592  & 124388      & 66              & 70      & 105.51                   & 3.02                   & 0.31            & 1.71           & 110.55          & 95.44            & 2.73  & 0.28 & 1.55            \\
			& art\_3\_12     & 55296  & 161856      & 72              & 76      & 133.58                   & 4.21                   & 0.39            & 2.06          & 140.24          & 95.25             & 3.01  & 0.27 & 1.47           \\
			& art\_3\_13     & 70304  & 206180      & 78              & 83      & 1153.20                  & 6.97                   & 1.52            & 2.84           & 1164.54         & 99.03           & 0.60  & 0.13 & 0.24             \\
			& art\_3\_14     & 87808  & 257936      & 84              & 88      & 2062.91                  & 9.09                   & 0.94             & 3.49         & 2076.43         & 99.35           & 0.44  & 0.05 & 0.17             \\
			& art\_3\_15     & 108000 & 317700      & 90              & 94      & 2240.17                  & 10.20                  & 0.77            & 4.11           & 2255.25         & 99.33          & 0.45  & 0.03 & 0.18              \\
			& art\_3\_16     & 131072 & 386048      & 96              & 100     & 971.23                   & 12.70                  & 0.56             & 5.75           & 990.23          & 98.08           & 1.28  & 0.06 & 0.58             \\
			& art\_3\_17     & 157216 & 463556      & 102             & 108     & 6068.14                  & 15.84                  & 4.81            & 0.43           & 6089.22         & 99.65           & 0.26  & 0.08 & 0.01             \\
			& art\_3\_18     & 186624 & 550800      & 108             & 112     & 5133.03                  & 16.57                  & 0.95             & 0.47           & 5151.01         & 99.65           & 0.32  & 0.02 & 0.01             \\
			& art\_3\_19     & 219488 & 648356      & 114             & 118     & 904.41                   & 18.84                  & 1.11            & 0.57           & 924.93          & 97.78        & 2.04  & 0.12 & 0.06                \\
			& art\_3\_20     & 256000 & 756800      & 120             & 124     & 11915.93                  & 30.30                   & 1.97            & 0.65           & 11948.85         & 99.72           & 0.25  & 0.02 & 0.01             \\
			& art\_4\_04     & 32768  & 120832      & 32              & 38      & 65.30                    & 2.23                   & 0.55            & 0.27           & 68.35           & 95.54        & 3.26  & 0.81 & 0.40                \\
			& art\_4\_05     & 80000  & 300000      & 40              & 46      & 232.23                   & 5.95                   & 0.49             & 0.67            & 239.34          & 97.03         & 2.49  & 0.21 & 0.28               \\
			& art\_4\_06     & 165888 & 628992      & 48              & 55      & 768.95                   & 16.57                  & 5.61            & 0.96            & 792.09          & 97.08         & 2.09  & 0.71 & 0.12               \\
			& art\_4\_07     & 307328 & 1174432     & 56              & 62      & 2151.37                  & 28.12                 & 4.53            & 1.06           & 2185.07         & 98.46          & 1.29  & 0.21 & 0.05              \\
			& art\_4\_08     & 524288 & 2015232     & 64              & 70      & 3373.92                  & 61.78                  & 10.88           & 1.60           & 3448.17         & 97.85         & 1.79  & 0.32 & 0.05               \\
			& art\_4\_09     & 839808 & 3242592     & 72              & 78      & 4293.87                   & 57.98                 & 4.95            & 2.07           & 4358.87          & 98.51          & 1.33  & 0.11 & 0.05              \\
			& seq\_40        & 164    & 164         & 164             & 164     & 0.04                     & 0.23                   & 0.00            & 1.47           & 1.75            & 2.54           & 13.19 & 0.01 & 84.26              \\
			& pparb\_2\_6    & 69632  & 321536      & 34              & 77      & 886.54                   & 25.51                  & 30.27           & 8.32           & 950.65          & 93.26          & 2.68  & 3.18 & 0.88              \\ \hline
			& AVERAGE        &        &             &                 &         &                             &                          &                   &                  &                    & 63.70                   & 20.69 & 0.28 &    15.33                                                                \\ \hline                                 
		\end{tabular}}
		\caption{TS statistics and CPU time for each decomposition step including the time spent to generate the regions}
		\label{tab:reg_gen_time}
	\vspace{-0.5cm}
\end{table*}

The set of SMs generated by Algorithm~\ref{alg:EC_set_generation} may be redundant, i.e., it may contain a subset of SMs which still define an ECTS. We describe a greedy search algorithm to obtain an irredundant set of SMs: we
order all the SMs by size and try to remove them one by one starting from the largest to the smallest, 
by checking that the union of the remaining regions satisfies \textit{excitation closure} and \textit{event effectiveness}. If excitation closure and event effectiveness are preserved, then the given SM can be removed. This algorithm is not optimal, because the removal of an SM may prevent the removal of a set of smaller SMs whose sum of places is greater than the number of places of the removed SM. However, this approach guarantees good performance having linear complexity in the number of SMs.

After the removal of the redundant SMs from the set shown in Fig.~\ref{fig:complete_decomposition_TS} only $SM_4$, $SM_5$, $SM_6$ and $SM_8$ are left.

\subsection{Merge between regions preserving the excitation closure}
\label{section:merge}

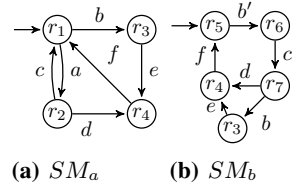
\begin{wrapfigure}{l}{0.43\linewidth}
	\centering
	\vspace{-0.3cm}
	\begin{subfigure}[t]{0.45\linewidth}
		\centering
		\scalebox{0.8}{
		\begin{tikzpicture}[->,>=stealth',shorten >=1pt,auto,node distance=1.4cm,semithick, initial text={}]
		\tikzstyle{every state}=[inner sep=1pt,minimum size=0pt]
		
		\node [initial,state] (r1) {$r_1$};
		\node [state] (r2)  [below of=r1] {$r_2$};
		\node [state] (r3)  [right of=r1] {$r_3$};
		\node [state] (r4) [right of=r2] {$r_4$};

		\path (r1) edge [bend left=10] node {$a$} (r2)
		(r1) edge node {$b$} (r3)
		(r2) edge [bend left=10] node {$c$} (r1)
		(r2) edge [below left] node {$d$} (r4)
		(r3) edge node {$e$} (r4)
		(r4) edge [above right] node {$f$} (r1)
		;
		\end{tikzpicture}}
		\caption{$SM_a$}
		\end{subfigure}
	~
	\hfill
	\begin{subfigure}[t]{0.45\linewidth}
		\centering
		\scalebox{0.8}{
		\begin{tikzpicture}[->,>=stealth',shorten >=1pt,auto,node distance=1cm,semithick, initial text={}]
		\tikzstyle{every state}=[]
		
		\node [state,inner sep=1pt,minimum size=0pt] (r3) {$r_3$};
		\node [state,inner sep=1pt,minimum size=0pt] (r7) [above right of=r3] {$r_7$};
		\node [state,inner sep=1pt,minimum size=0pt] (r6) [above of=r7] {$r_6$};
		\node [initial,state,inner sep=1pt,minimum size=0pt] (r5)  [left of=r6] {$r_5$};
		\node [state,inner sep=1pt,minimum size=0pt] (r4)  [below of=r5] {$r_4$};

		\path (r3) edge [left] node {$e$} (r4)
		(r4) edge node {$f$} (r5)
		(r5) edge node {$b'$} (r6)
		(r6) edge node {$c$} (r7)
		(r7) edge node {$b$} (r3)
		(r7) edge [above] node {$d$} (r4)
		;
		\end{tikzpicture}}
		\caption{$SM_b$}
		\end{subfigure}
	\caption{Initial SMs.}
	\label{fig:decomposition_ECTS3}
	\vspace{-0.3cm}
	\end{wrapfigure}

The third step of the procedure merges pairs of regions with the objective to minimize the size of the sets of SMs: edges carrying labels are removed, and by consequence the two nodes connected to them are merged decreasing their number. E.g., both the SMs in Fig.~\ref{fig:decomposition_ECTS3} (obtained from a TS different from the one in Fig.~\ref{fig:TS}) contain an instance of label $e$ connected by regions $r_3$ and $r_4$. This means that an edge carrying label $e$ can be removed in one of the SMs. The result of removing the edge with label $e$ in $\sm_b$ and merging the regions $r_3$ and $r_4$ replacing them with the region $r_{34}$ is shown in Fig.~\ref{fig:merge_ECTS3}.

\begin{wrapfigure}{l}{0.5\linewidth}
	\centering
	\vspace{-0.3cm}
	\begin{subfigure}[t]{0.45\linewidth}
		\centering
		\scalebox{0.8}{
		\begin{tikzpicture}[->,>=stealth',shorten >=1pt,auto,node distance=1.4cm,semithick, initial text={}]
		\tikzstyle{every state}=[]
		
		\node [initial,state,inner sep=1pt,minimum size=0pt] (r1) {$r_1$};
		\node [state,inner sep=1pt,minimum size=0pt] (r2)  [below of=r1] {$r_2$};
		\node [state,inner sep=1pt,minimum size=0pt] (r3)  [right of=r1] {$r_3$};
		\node [state,inner sep=1pt,minimum size=0pt] (r4) [right of=r2] {$r_4$};

		\path (r1) edge [bend left=10] node {$a$} (r2)
		(r1) edge node {$b$} (r3)
		(r2) edge [bend left=10] node {$c$} (r1)
		(r2) edge [below] node {$d$} (r4)
		(r3) edge node {$e$} (r4)
		(r4) edge [above right] node {$f$} (r1)
		;
		\end{tikzpicture}}
		\caption{$SM_a$}
	\end{subfigure}
	~
	\begin{subfigure}[t]{.45\linewidth}
		\centering
		\scalebox{0.8}{
		\begin{tikzpicture}[->,>=stealth',shorten >=1pt,auto,node distance=1cm,semithick, initial text={}]
		\tikzstyle{every state}=[]
		
		\node [state,inner sep=1pt,minimum size=0pt] (r34) {$r_{34}$};
		\node [initial,state,inner sep=1pt,minimum size=0pt] (r5)  [above of=r34] {$r_5$};
		\node [state,inner sep=1pt,minimum size=0pt] (r6) [right of=r5] {$r_6$};
		\node [state,inner sep=1pt,minimum size=0pt] (r7) [below of=r6] {$r_7$};
		
		\path 
		(r34) edge node {$f$} (r5)
		(r5) edge node {$b'$} (r6)
		(r6) edge node {$c$} (r7)
		(r7) edge [above, bend right=20] node {$b$} (r34)
		(r7) edge [below, bend left=20] node {$d$} (r34)
		;
		\end{tikzpicture}}
		\caption{$SM_b$}
	\end{subfigure}
	\caption{SMs of Fig.~\ref{fig:decomposition_ECTS3} after the removal of label $e$ in $SM_b$.}
	\label{fig:merge_ECTS3}
	\vspace{-0.3cm}
\end{wrapfigure}
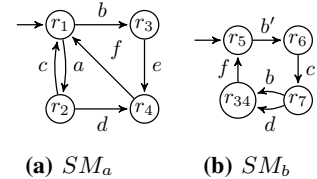
All instances of a region except one can be removed, because removing all of them would change the set of regions used for checking the excitation closure-property, whereas keeping at least one guarantees the preservation of the property.

We formulated the merging problem as solving an instance of SAT. We will skip the exact SAT clause encoding due to lack of space. 
According to the SAT solution, the SMs are restructured by removing
arcs and nodes to be deleted and adding merged nodes, and redirecting arcs as appropriate. In the running example, in $\sm_b$ we merge the nodes $r_3$, $r_4$ into 
node $r_{34}$, remove the edge labeled $e$ between the deleted nodes $r_3$ and $r_4$,
and redirect to $r_{34}$ the edges pointing to $r_3$ or $r_4$.

Instead, none of the four SMs surviving the irredundancy step from Fig.~\ref{fig:TS} is further minimized by the merging step.
\section{Composition of SMs and equivalence to the original TS} 
\label{section:relations}

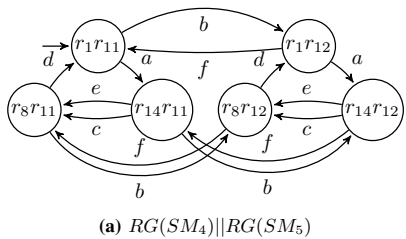
\begin{wrapfigure}{r}{0.56\linewidth}
	\centering
	\vspace{-0.3cm}
	\scalebox{0.8}{
	\begin{minipage}{\linewidth}
	\begin{subfigure}[t]{\linewidth}
		\hspace{-0.9cm}
	\begin{tikzpicture}[->,>=stealth',shorten >=1pt,auto,node distance=1.5cm,semithick, initial text={}]
	\tikzstyle{every state}=[]
	
	\node [initial,state,inner sep=1pt,minimum size=0pt] (r1r11) {$r_1r_{11}$};
	\node [state,inner sep=1pt,minimum size=0pt] (r8r11)  [below left of=r1r11] {$r_8r_{11}$};
	\node [state,inner sep=1pt,minimum size=0pt] (r14r11)  [below right of=r1r11] {$r_{14}r_{11}$};
	\node [state,inner sep=1pt,minimum size=0pt] (r1r12) [right of=r1r11, xshift=2cm]{$r_1r_{12}$};
	\node [state,inner sep=1pt,minimum size=0pt] (r8r12)  [below left of=r1r12] {$r_8r_{12}$};
	\node [state,inner sep=1pt,minimum size=0pt] (r14r12)  [below right of=r1r12] {$r_{14}r_{12}$};
	
	\path (r1r11) edge [bend left=10] node {$a$} (r14r11)
	(r8r11) edge [bend left=10] node {$d$} (r1r11)
	(r14r11) edge [bend left=10] node {$c$} (r8r11)
	(r14r11) edge [bend right=10, above] node {$e$} (r8r11)
	(r1r12) edge [bend left=10] node {$a$} (r14r12)
	(r8r12) edge [bend left=10] node {$d$} (r1r12)
	(r14r12) edge [bend left=10] node {$c$} (r8r12)
	(r14r12) edge [bend right=10, above] node {$e$} (r8r12)

	(r1r11) edge [bend left=30, below] node {$b$} (r1r12)
	(r8r11) edge [bend right=55, below] node {$b$} (r8r12)
	(r14r11) edge [bend right=50, below] node {$b$} (r14r12)
	(r1r12) edge [bend left=5, below] node {$f$} (r1r11)
	(r8r12) edge [bend left=45, above] node {$f$} (r8r11)
	(r14r12) edge [bend left=40, above] node {$f$} (r14r11)
	;
	\end{tikzpicture}
	\caption{$RG(SM_4)||RG(SM_5)$}
	\end{subfigure}
	\end{minipage}
	}
	\caption{Composition between $RG(SM_4)$ and $RG(SM_5)$ of Fig.~\ref{fig:complete_decomposition_TS}}
	\label{fig:composition_SM_45}
	\vspace{-0.3cm}
\end{wrapfigure}

Intuitively, the SMs derived from an LTS interact running in parallel with the same rules of the synchronous product  
of transition systems (see Def.~\ref{def:parallel_composition}). Indeed, if we interpret the reachability graphs of the SMs as LTSs and execute the synchronous product
deriving a single LTS which models the interaction of the SMs, it turns out that the result of the composition is equivalent to the original LTS, as proved in the appendix.
E.g., consider the composition of reachability graphs of SMs $\sm_4$ and $\sm_5$ in Fig.~\ref{fig:composition_SM_45}, it generates a superset of behaviors of the original LTS  in Fig.~\ref{fig:TS}: it produces the sequence ``$\mathrm{acbdaefd}$'' which is in the original LTS, but also new behaviors, like the sequences starting by the event $b$ (e.g. ``$\mathrm{bacfd}$''), 
which are not in the original LTS because some constraints of the original LTS are missing; indeed, these two SMs are not enough to satisfy the excitation-closure property, whereas event effectiveness is satisfied by them because all events are included in the composition. 
The composition of SMs can exhibit these hidden behaviors by including new regions: e.g., the composition of $\sm_4$ with $\sm_5$ includes two new regions $r_{11}$ and $r_{12}$ so that the events $b$ and $f$ show up in the composition.

The equivalence between an ECTS and the derived set of SMs is proved by defining a bisimulation between
the original TS and the synchronous product 
of the reachability graphs of the derived state machines 
$\rgraph{\sm_1} || \rgraph{\sm_2} || \dots$ $|| \rgraph{\sm_n}$,
denoted by $||_{i=1,\ldots,n} \rgraph{\sm_i}$.

\begin{theorem} \label{thm:equivalence-n}
Given an excitation-closed set \mbox{$\{\sm_1, \dots, \sm_n \}$} of SMs derived from the ECTS $\textit{TS}$, there is a bisimulation $B$ such that
\mbox{$\textit{TS} \sim_B ||_{i=1,\ldots,n} \rgraph{\sm_i}$}.
\end{theorem}
\begin{proof}
See the appendix.
\end{proof}

Theorem \ref{thm:equivalence-n} 
states that, given a set of SMs, the excitation closure and event effectiveness of the union of their regions is a necessary and sufficient condition to guarantee that their synchronous product is equivalent to the original TS. 

\section{Experimental results} \label{section:experiments}

\begin{table*}
\centering
\resizebox{0.9\textwidth}{!}{
\begin{tabular}{|l|rr|rrr|rrr|S[table-format=3.0]S[table-format=4.0]c|S[table-format=2.0]S[table-format=2.2]S[table-format=2.2]S[table-format=2.0]S[table-format=2.0]|}
\hline
\multicolumn{1}{|c|}{}                       & \multicolumn{11}{c|}{\bf Size comparison}    & \multicolumn{5}{c|}{\bf SM details}      \\ \hline
\multicolumn{1}{|c|}{\multirow{3}{*}{\bf Input}} & \multicolumn{2}{c|}{\multirow{2}{*}{\bf TS}}            & \multicolumn{3}{c|}{\multirow{2}{*}{\bf PN}}                               & \multicolumn{3}{c|}{\multirow{2}{*}{\bf PN*\footnotemark
}}                                                & \multicolumn{3}{c|}{\bf Synchronizing}                                     & \multicolumn{1}{c}{\multirow{3}{*}{\begin{tabular}[c]{@{}c@{}}\bf Number \\ \bf of SMs\end{tabular}}} & \multicolumn{1}{c}{\bf Avg.}   & \multicolumn{1}{c}{\bf Avg.}     & \multicolumn{1}{c}{\bf Places}  & \multicolumn{1}{c|}{\bf Alphabet} \\
\multicolumn{1}{|c|}{}                       & \multicolumn{2}{c|}{}                               & \multicolumn{3}{c|}{}                                                  & \multicolumn{3}{c|}{}                                          & \multicolumn{3}{c|}{\bf SMs}                                               & \multicolumn{1}{c}{}                                                                          & \multicolumn{1}{c}{\bf places} & \multicolumn{1}{c}{\bf alphabet} & \multicolumn{1}{c}{\bf largest} & \multicolumn{1}{c|}{\bf largest}  \\ \cline{2-12}
\multicolumn{1}{|c|}{}                       & \multicolumn{1}{c}{\bf States} & \multicolumn{1}{c|}{\bf T} & \multicolumn{1}{c}{\bf P} & \multicolumn{1}{c}{\bf T} & \multicolumn{1}{c|}{\bf C} & \multicolumn{1}{c}{\bf P} & \multicolumn{1}{c}{\bf T} & \multicolumn{1}{c|}{\bf C} & \multicolumn{1}{c}{\bf P} & \multicolumn{1}{c}{\bf T} & \multicolumn{1}{c|}{\bf C} & \multicolumn{1}{c}{}                                                                          & \multicolumn{1}{c}{\bf per SM} & \multicolumn{1}{c}{\bf per SM}   & \multicolumn{1}{c}{\bf SM}      & \multicolumn{1}{c|}{\bf SM}       \\ \hline
alloc-outbound                             & 21                         & 18                    & 14                    & 14                    & 3                     & 17                    & 18                    & 0                     & 17                    & 21                    & 0                     & 2                                                                                             & 8.50                       & 10.50                        & 10                          & 11                           \\
clock                                      & 10                         & 10                    & 8                     & 5                     & 4                     & 10                    & 10                    & 0                     & 11                    & 15                    & 0                     & 3                                                                                             & 3.67                       & 5.00                         & 4                           & 4                            \\
dff                                        & 20                         & 24                    & 13                    & 14                    & 21                    & 20                    & 20                    & 0                     & 25                    & 41                    & 0                     & 3                                                                                             & 8.33                       & 13.33                        & 13                          & 7                            \\
espinalt                                   & 27                         & 31                    & 22                    & 20                    & 5                     & 27                    & 25                    & 1                     & 29                    & 32                    & 0                     & 3                                                                                             & 9.33                       & 11.00                        & 11                          & 13                           \\
fair\_arb                                  & 13                         & 20                    & 11                    & 10                    & 4                     & 11                    & 10                    & 4                     & 12                    & 18                    & 0                     & 2                                                                                             & 6.00                       & 9.00                         & 6                           & 6                            \\
future                                     & 36                         & 44                    & 18                    & 16                    & 1                     & 30                    & 28                    & 0                     & 21                    & 22                    & 0                     & 3                                                                                             & 7.00                       & 7.33                         & 13                          & 14                           \\
intel\_div3                                & 8                          & 8                     & 7                     & 5                     & 2                     & 8                     & 8                     & 0                     & 10                    & 11                    & 0                     & 2                                                                                             & 5.00                       & 5.50                         & 6                           & 4                            \\
intel\_edge                                & 28                         & 36                    & 11                    & 15                    & 22                    & 21                    & 30                    & 56                    & 35                    & 68                    & 1                     & 4                                                                                             & 8.50                       & 16.75                        & 13                          & 6                            \\
isend                                      & 53                         & 66                    & 25                    & 27                    & 106                   & 54                    & 43                    & 5                     & 80                    & 138                   & 4                     & 13                                                                                            & 6.31                       & 11.85                        & 12                          & 11                           \\
lin\_edac93                                & 20                         & 28                    & 10                    & 8                     & 1                     & 14                    & 12                    & 0                     & 13                    & 14                    & 0                     & 3                                                                                             & 4.33                       & 4.67                         & 5                           & 6                            \\
master-read                                & 8932                       & 36226                 & 33                    & 26                    & 0                     & 33                    & 26                    & 0                     & 38                    & 38                    & 0                     & 8                                                                                             & 4.75                       & 4.75                         & 10                          & 10                           \\
pe-rcv-ifc                                 & 46                         & 62                    & 23                    & 20                    & 96                    & 43                    & 37                    & 13                    & 39                    & 57                    & 2                     & 2                                                                                             & 19.00                      & 28.50                        & 21                          & 13                           \\
pulse                                      & 12                         & 12                    & 7                     & 6                     & 2                     & 12                    & 12                    & 0                     & 7                     & 10                    & 0                     & 2                                                                                             & 3.50                       & 5.00                         & 3                           & 6                            \\
rcv-setup                                  & 14                         & 17                    & 10                    & 10                    & 5                     & 14                    & 14                    & 4                     & 12                    & 14                    & 0                     & 2                                                                                             & 6.00                       & 7.00                         & 9                           & 10                           \\
vme\_read                                  & 255                        & 668                   & 38                    & 29                    & 18                    & 41                    & 32                    & 2                     & 50                    & 67                    & 1                     & 9                                                                                             & 6.11                       & 7.67                         & 12                          & 13                           \\
vme\_write                                 & 821                        & 2907                  & 46                    & 33                    & 31                    & 49                    & 36                    & 6                     & 57                    & 74                    & 1                     & 11                                                                                            & 6.18                       & 7.36                         & 9                           & 11                           \\ \hline
\end{tabular}
}
	\caption{Number of places (P), transitions (T) and arc crossings (C) of the original transition systems vs. derived Petri nets vs. product of SMs and SM details.
	}
	\label{tab:comparison}
\vspace{-0.5cm}

\end{table*}

We implemented the procedure described in Sec.~\ref{section:decomposition} and performed experiments on an Intel core running at 2.80GHz with 16GB of RAM. Our software is written in C++ and uses \textit{PBLib}~\cite{pblib.sat2015} for the resolution of SAT. The resolution of the MIS problem is performed by the \textit{NetworkX} library~\cite{team2014networkx}. 
For our tests, we used two sets of benchmarks: the first set (the same as in~\cite{regions}), with smaller transition systems is listed in the first rows of Table~\ref{tab:reg_gen_time} and denoted as ``Small-sized set"; the second one containing large transition systems is listed in the second part of Table~\ref{tab:reg_gen_time}, 
denoted as ``Large-sized set".

Table~\ref{tab:reg_gen_time} shows the absolute and relative runtimes of the steps of the flow: region generation, decomposition into SMs, irredundancy, place merging. The generation of minimal regions is the dominating operation taking more than 60\% of the overall time spent; it is exponential in the number of events and with the increase of the input dimensions it becomes the bottleneck shadowing the remaining computations. 
However it is still possible to decompose quite large transition systems with about $10^6$ states and $3\cdot 10^6$ transitions.

Table~\ref{tab:comparison} compares the states and transitions of transition systems vs. the places/transitions/crossing arcs of the Petri nets derived by Petrify~\cite{cortadella1997petrify} (columns under PN), and vs. our product of state machines for the first benchmark set. The number of crossing arcs is reported by the {\em dot} algorithm of {\em graphviz}~\cite{crossings} and can be considered as a metric of structural simplicity of the model (i.e., fewer crossings implies a simpler structure). Our results from synchronized state machines have similar sizes compared to those from Petri nets, but they have fewer crossings, which is a significant advantage in supporting a visual representation for ``large systems". Therefore the plots, in a two-dimensional graphical representation of synchronizing SMs, are substantially more \textit{readable} than the ones of Petri nets: see the inputs \textit{intel\_edge} and \textit{pe-rcv-ifc} witnessing that peaks of edge crossings are avoided. 
The example \textit{master-read} instead is an impressive case of how our decomposition tames the state explosion of the original transition system derived from a highly concurrent environment, since from 8932 states we go down to 8 SMs with an average number of 5 states each.

\begin{table*}
	\centering
	\resizebox{\textwidth}{!}{
	\begin{tabular}{|l|S[table-format=4.2]rS[table-format=2.2]|rrr|rrr|r|}
	\hline
		{\bf Input}          & \multicolumn{1}{c}{\begin{tabular}[c]{@{}c@{}}{ \bf Decomposition} \\ { \bf{ [}s{]}}\end{tabular}} & \multicolumn{1}{c}{\begin{tabular}[c]{@{}c@{}}{\bf Greedy} \\ { \bf{ [}s{]}}\end{tabular}} & \multicolumn{1}{c|}{\begin{tabular}[c]{@{}c@{}}{\bf Merge} \\ { \bf{ [}s{]}}\end{tabular}} & \multicolumn{1}{c}{\begin{tabular}[c]{@{}c@{}}{ \bf States}\\ {\bf after} \\ {\bf decomposition}\end{tabular}} & \multicolumn{1}{c}{\begin{tabular}[c]{@{}c@{}}{\bf States}\\ {\bf after} \\ {\bf greedy}\end{tabular}} & \multicolumn{1}{c|}{\begin{tabular}[c]{@{}c@{}}{\bf States}\\ {\bf after} \\ {\bf merge}\end{tabular}} & \multicolumn{1}{c}{\begin{tabular}[c]{@{}c@{}}{\bf Trans.}\\ {\bf after} \\ {\bf decomposition}\end{tabular}} & \multicolumn{1}{c}{\begin{tabular}[c]{@{}c@{}}{\bf Trans.}\\ {\bf after }\\ {\bf greedy}\end{tabular}} & \multicolumn{1}{c|}{\begin{tabular}[c]{@{}c@{}}{\bf Trans.}\\ {\bf after} \\ {\bf merge}\end{tabular}} & \multicolumn{1}{c|}{\begin{tabular}[c]{@{}c@{}}{\bf Number}\\ {\bf of regions }\\{\bf TS}\end{tabular}} \\ \hline
		alloc-outbound & 14.01             & 0.0009   & 0.06   & 42                                                                                        & 21                                                                                 & 17                                                                                & 50                                                                                        & 25                                                                                 & 21                                                                                & 15                                                                              \\
		clock          & 0.55            & 0.0003   & 0.02  & 18                                                                                        & 14                                                                                 & 11                                                                                & 22                                                                                        & 18                                                                                 & 15                                                                                & 11                                                                              \\
		fair\_arb      & 0.58            & 0.0007   & 0.03   & 24                                                                                        & 12                                                                                 & 12                                                                                & 36                                                                                        & 18                                                                                 & 18                                                                                & 11                                                                              \\
		future         & 1881.00              & 0.0012   & 0.10   & 41                                                                                        & 29                                                                                 & 22                                                                                & 43                                                                                        & 30                                                                                 & 23                                                                                & 19                                                                              \\
		intel\_div3    & 0.21            & 0.0001   & 0.01  & 12                                                                                        & 12                                                                                 & 10                                                                                & 13                                                                                        & 13                                                                                 & 11                                                                                & 8                                                                               \\
		lin\_edac93    & 0.33            & 0.0002   & 0.01   & 13                                                                                        & 13                                                                                 & 13                                                                                & 14                                                                                        & 14                                                                                 & 14                                                                                & 10                                                                              \\
		pulse          & 0.20            & 0.0000   & 0.01  & 7                                                                                         & 7                                                                                  & 7                                                                                 & 10                                                                                        & 10                                                                                 & 10                                                                                & 7                                                                               \\
		rcv-setup      & 0.36            & 0.0002   & 0.04   & 18                                                                                        & 18                                                                                 & 12                                                                                & 22                                                                                        & 22                                                                                 & 14                                                                                & 11  \\ \hline                                                                         
	\end{tabular}}
	\caption{CPU time and results of the exact decomposition algorithm} \label{tab:exact_decomposition}
\end{table*}

We implemented also an exact search of all SMs derived from the original TS, to gauge our heuristics, when it is possible to find an exact solution.
We compare the times taken by the exact and heuristic SM generation steps: the exponential behaviour of the exact algorithm makes it hardly affordable for about 15 regions and run out of 16GB of memory for more than 20 regions (Table~\ref{tab:exact_decomposition}).
Instead, the approximate algorithms presented in Sec.~\ref{section:decomposition} can handle very large transition systems. Even though the result is not guaranteed to be a minimum one, the irredundancy procedure guarantees a form of minimality, yielding a compact representation that avoids state explosion and exhibits concurrency explicitly. 

\footnotetext{PN* is a representation of the PN after splitting disconnected ERs, thus producing multiple labels (transitions) for the same event. This results in a PN with more  transitions and a simpler structure.}

\section{Conclusions} \label{section:conclusions}

In this paper we described a new method for the decomposition of transition systems. Our experimental results demonstrate that the decomposition algorithm can be run on transition systems with up to one million states, therefore, it is suitable to handle real cases. Since the generation of minimal regions is currently a computational bottleneck, future work will address this limitation, while it will leverage the improvements in efficiency of last-generation MIS and SAT solvers.

As future work, we want to apply this decomposition paradigm to process mining. Rather than synthesizing intricate ``spaghetti'' Petri nets from logs, we aim at distilling loosely coupled concurrent threads (SMs) that can be easily visualized, analyzed and optimized individually, while preserving the synchronization with the other threads. Optionally, a new Petri net can be obtained by composing back the optimized threads and imposing some structural constraints, e.g., to be a Free-Choice Petri net, thus providing a tight approximation of the original behavior with a simpler structure.

\bibliography{bibliography}
\bibliographystyle{style/IEEEtran}

\appendix
\section{Appendix}

\subsection{Proof of theorem \ref{thm:equivalence-n}} \label{section:equivalence_proof}

The equivalence between an ECTS and the derived set of SMs is proved by defining a bisimulation between
the original TS, defined as \mbox{$\textit{TS} = (S, E, T, s_{0})$}, and the synchronous product
of the reachability graphs of the derived state machines 
\mbox{$\rgraph{\sm_1} || \rgraph{\sm_2} || \dots || \rgraph{\sm_n}$},
denoted by \mbox{$||_{i=1, \dots, n} \rgraph{\sm_i} = (S_{||},$ $E, T_{||},$} $s_{0,||})$.
Notice that each \mbox{$RG(SM_i) = (R_i,E_i,T_i,r_{0,i})$}, with
$T_i \subseteq R_i \times E_i \times R_i$,
is defined on a subset $E_i$ of events of $\textit{TS}$, its states $r_i$ correspond to regions of the states of $\textit{TS}$, and the initial state is a region $r_{0,i}$ containing the initial state of $\textit{TS}$.
To prove the existence of a bisimulation, we require that the union of $\rgraph{\sm_i}$ satisfies ECTS, where event-effectiveness guarantees that \mbox{$\cup E_i = E$}, and excitation-closure guarantees that the two transition systems simulate each other, i.e., the transition relations allow to match each other's moves.

\begin{proof}

We define the binary relation $B$ as follows: 
\begin{equation}
    (s_j, (r_{j,1}, r_{j,2}, \dots, r_{j,n})) \in B \iff s_j \in \bigcap_{i=1}^n r_{j,i},
\end{equation}
where \mbox{$s_j \in S$} and \mbox{$r_{j,i}\in R_i$}, for $i\in\{1,\ldots,n\}$.

Notice that writing $(s_j, (r_{j,1}, r_{j,2}, \dots, r_{j,n})) \in B \iff \{s_j\} = \bigcap_{i=1}^n r_{j,i}$ would be wrong, because the intersection of regions could have two or more bisimilar (i.e., behaviourally equivalent) states, as in the TS $\arc{s_0}{a}{}\arc{s_1}{b}{}\arc{s_2}{a}{}\arc{s_3}{b}{s_0}$.

A region $r_{j,i}$ may appear in two or more sets of regions $R_i$.
Now we prove that $B$ is a bisimulation in three steps:
\begin{enumerate}
    \item $(s_0,(r_{0,1},r_{0,2}, \dots,r_{0,n})) \in B$.
    \item If \mbox{$(s_j, (r_{j,1}, r_{j,2}, \dots, r_{j,n})) \in B$} and 
$(s_j,e,s_k) \in T$, then there is $(r_{k,1}, r_{k,2},$ $\dots, r_{k,n}) \in S_{||}$ such that $((r_{j,1}, r_{j,2}, \dots, r_{j,n}),e,(r_{k,1}, r_{k,2}, \dots, r_{k,n})) \in T_{||}$ and
$(s_k,(r_{k,1}, r_{k,2}, \dots, r_{k,n})) \in B$.
\item  If $(s_j, (r_{j,1}, r_{j,2}, \dots, r_{j,n})) \in B$ and
\mbox{$((r_{j,1}, r_{j,2}, \dots, r_{j,n}), e,(r_{k,1}, r_{k,2}, \dots, r_{k,n})) \in T_{||}$}, then there is $s_k \in S$ such that $(s_j,e,s_k) \in T$ and \mbox{$(s_k,(r_{k,1}, r_{k,2}, \dots, r_{k,n})) \in B$.}
\end{enumerate}

\medskip \noindent
Let us now proceed with the proofs.
\begin{enumerate}
\item 
Since $\textit{TS}$ has a unique initial state $s_0$, each state machine $\sm_i$ has exactly one initial region $r_{0,i}$ such that $s_0 \in r_{0,i}$ because all the regions of an SM are disjoint. Therefore, $s_0 \in \bigcap_{i=1}^n r_{0, i}$ and we have that \mbox{$(s_0,(r_{0,1},r_{0,2}, \dots,r_{0,n})) \in B$.}

\item Since $(s_j,e,s_k) \in T$ and 
$(s_j, (r_{j,1}, r_{j,2}, \dots, r_{j,n})) \in B$, then $s_j\in\bigcap_{i=1}^n r_{j,i}$.
Now we will prove that there is $s_k$ such that \mbox{$s_k\in\bigcap_{i=1}^n r_{k,i}$}, so that we can have $(s_k,(r_{k,1}, r_{k,2},$ $\dots, r_{k,n})) \in B$.

Since $e$ is enabled in $s_j$, none of the $r_{j,i}$'s can be a post-region of $e$. If one $r_{j,i}$ in $\{r_{j,1},\ldots,r_{j,n}\}$ would be a post-region,
then $s_j\not\in \bigcap_{i=1}^n r_{j,i}$.
Therefore, the following three cases can be distinguished for each 
\mbox{$r_{j,i} \in \{r_{j,1}, r_{j,2}, \dots, r_{j,n}\}$}:

\begin{itemize}
    \item $e$ is not an event of $\sm_i$. Thus, $r_{k,i}= r_{j,i}$. 
    \item $e$ is an event of $\sm_i$ and $r_{j,i}$ is a no-cross region for $e$. Thus, $r_{k,i}= r_{j,i}$.
    \item $e$ is an event of $\sm_i$ and $r_{j,i}$ is a pre-region of $e$. Thus, $r_{k,i}\neq r_{j,i}$ is a post-region of $e$.
\end{itemize}

For the first and second cases, $\sm_i$ will not change state and $\textit{TS}$ will not change region when moving from $s_j$ to $s_k$. Therefore, \mbox{$s_k\in r_{j,i} = r_{k,i}$}.

For the third case, $e$ will exit $r_{j,i}$ and will enter $r_{k,i}$ in $\textit{TS}$, which means that \mbox{$s_k\in r_{k,i}$}. Therefore,
\mbox{$((r_{j,1}, r_{j,2}, \dots, r_{j,n}),e,(r_{k,1}, r_{k,2}, \dots, r_{k,n})) \in T_{||}$}.

For all cases we have that $s_k\in r_{k,i}$ and therefore $s_k\in\bigcap_{i=1}^n r_{k,i}$.

\item 
Since $(s_j, (r_{j,1}, r_{j,2}, \dots, r_{j,n})) \in B$, it holds that 
\mbox{$s_j \in \bigcap_{i=1}^n r_{j,i}$}. Given the existence
of the transition \mbox{$((r_{j,1}, r_{j,2}, \dots, r_{j,n}), e,(r_{k,1}, r_{k,2}, \dots, r_{k,n}))$}, and knowing that the EC property holds, we know that
\mbox{$s_j \in \bigcap_{i=1}^{n} r_{j,i} \subseteq \es{e}$}.
The latter inequality holds because by Th.~\ref{thm:regions_to_SM} we have 1) $\forall i, i =1, \dots, n$, label $e$ appears once in $\sm_i$ or it does not appear, and 2) $\forall i, i =1, \dots, n$, if label $e$ appears in $\sm_i$ then $r_{j,i} \in (\preregion{e} \cap R)$, by which $\bigcup_{i=1}^{n} r_{j,i} \supseteq \bigcup_{r \in (\preregion{e} \cap R)} r$ and so by complementation $\bigcap_{i=1}^{n} r_{j,i} \subseteq \bigcap_{r \in (\preregion{e} \cap R)} r = \es{e}$.

Therefore, there is $s_k$ such that \mbox{$(s_j,e,s_k) \in T$}. We can also see that $s_k \in \bigcap_{i=1}^n r_{k,i}$, using the same reasoning as in step 2, since all the pre-regions $r_{j,i}$ of $e$ in $\{r_{j,1},\ldots,r_{j,n}\}$ are exited by entering $r_{k,i}$, whereas the no-crossing regions remain the same. We can then conclude that 
\mbox{$(s_k,(r_{k,1}, r_{k,2}, \dots, r_{k,n})) \in B$.}

\end{enumerate}
\vspace{-0.5cm}
\end{proof}

\end{document}